\documentclass[runningheads]{llncs}

\newcommand{\Oh}[1]
    {\ensuremath{\mathcal{O} \hspace{-.5ex} \left( {#1} \right)}}

\begin{document}

\title{Minimax Trees in Linear Time}
\author{Pawe\l\ Gawrychowski\inst{1}
    \and Travis Gagie\inst{2}\fnmsep\thanks
    {This paper was written while the second author was at the University of Eastern Piedmont, Italy, supported by Italy-Israel FIRB Project ``Pattern Discovery Algorithms in Discrete Structures, with Applications to Bioinformatics''.}}
\authorrunning{P. Gawrychowski and T. Gagie}
\institute{Institute of Computer Science\\
    University of Wroclaw, Poland\\
    \email{gawry1@gmail.com}\\\mbox{}\\
    \and Research Group in Genome Informatics\\
    University of Bielefeld, Germany\\
    \email{travis.gagie@gmail.com}}
\maketitle

\begin{abstract}
A minimax tree is similar to a Huffman tree except that, instead of minimizing the weighted average of the leaves' depths, it minimizes the maximum of any leaf's weight plus its depth.  Golumbic (1976) introduced minimax trees and gave a Huffman-like, $\Oh{n \log n}$-time algorithm for building them.  Drmota and Szpankowski (2002) gave another $\Oh{n \log n}$-time algorithm, which checks the Kraft Inequality in each step of a binary search.  In this paper we show how Drmota and Szpankowski's algorithm can be made to run in linear time on a word RAM with \(\Omega (\log n)\)-bit words.  We also discuss how our solution applies to problems in data compression, group testing and circuit design.
\end{abstract}

\section{Introduction} \label{sec:intro}

In a minimax tree for a multiset \(W = \{w_1, \ldots, w_n\}\) of weights, each leaf has a weight $w_i$, each internal node has weight equal to the maximum of its children's weights plus 1, and the weight of the root is as small as possible.  In other words, if $\ell_i$ is the depth of the leaf with weight $w_i$, then \(\max_i \{w_i + \ell_i\}\) is minimized.  The weight of the root is called the minimax cost of $W$, denoted \(M (W)\).  Golumbic~\cite{Gol76} showed that if we modify Huffman's algorithm~\cite{Huf52} to repeatedly replace the two nodes with smallest weights $w_i$ and $w_j$, by a node with weight \(\max (w_i, w_j) + 1\) instead of \(w_i + w_j\), then it builds a minimax tree instead of a Huffman tree.  Like Huffman's algorithm, it takes $\Oh{n \log n}$ time and can build trees of any degree.  Our results in this paper also generalize to higher degrees and larger code alphabets but, for the sake of simplicity, we henceforth consider only binary trees and alphabets.  Golumbic, Parker~\cite{Par79} and Hoover, Klawe and Pippenger~\cite{HKP84} showed how to use Golumbic's algorithm to restrict circuits' fan-in and fan-out without greatly increasing their sizes or depths.  Drmota and Szpankowski~\cite{DS02,DS04} pointed out that, if \(P = p_1, \ldots, p_n\) is a probability distribution and each \(w_i = \log (1 / p_i)\), then a minimax tree for $W$ is the code-tree for a prefix code with minimum maximum pointwise redundancy with respect to $P$.  (As we are considering only binary trees in this paper, by $\log$ we always mean $\log_2$.)  They gave another $\Oh{n \log n}$-time algorithm for building minimax trees and, by analyzing it, proved bounds on the redundancy of arithmetic coding, which Baer~\cite{Bae08} recently improved by analyzing Golumbic's algorithm.  Drmota and Szpankowski start with a \mbox{Shannon} code~\cite{Sha48} for $P$, in which the codeword for the $i$th character has length \(\lceil \log (1 / p_i) \rceil\), for each $i$; they sort the logarithms by their fractional parts, i.e., \(\log (1 / p_1) - \lfloor \log (1 / p_1) \rfloor, \ldots, \log (1 / p_n) - \lfloor \log (1 / p_n) \rfloor\); and they use binary search to find the largest value $x$ such that \(\lceil \log (1 / p_1) - x \rceil, \ldots, \lceil \log (1 / p_n) - x \rceil\) obey the Kraft Inequality~\cite{Kra49}.  In a previous paper~\cite{Gag04} (see also~\cite{Gag07,KN??}) we noted that minimax trees built with Golumbic's algorithm have the same Sibling Property~\cite{Fal73,Gal78} as Huffman trees, and turned the Faller-Gallager-Knuth algorithm~\cite{Knu85} for dynamic Huffman coding into an algorithm for dynamic Shannon coding.  Intriguingly, although static Huffman coding is optimal and static Shannon coding is not, dynamic Shannon coding has a better worst-case bound than dynamic Huffman coding does.

Hu, Kleitman and Tamaki~\cite{HKT79} gave an $\Oh{n \log n}$-time algorithm for building alphabetic minimax trees, in which the leaves' weights, from left to right, must be in the given order.  Kirkpatrick and Klawe~\cite{KK85} and Coppersmith, Klawe and Pippenger~\cite{CKP86} gave an algorithm (or, more precisely, two algorithms that are equivalent when trees are binary) that builds an alphabetic minimax tree for integer weights in $\Oh{n}$ time, and showed how to use it to restrict circuits' fan-in and fan-out without greatly increasing their sizes or depths and without changing the numbers of edge crossings (and, thus, preserving planarity).  Kirkpatrick and Klawe also showed how to combine their algorithm with binary search in order to build alphabet minimax trees for real weights in $\Oh{n \log n}$ time.  We note that, if their algorithm for integer weights is viewed as an alphabetic analogue of the Kraft Inequality --- as it was by Yeung~\cite{Yeu91} and Nakatsu~\cite{Nak91}, who independently rediscovered it --- then their algorithm for real weights is the alphabetic analogue of Drmota and Szpankowski's.  Kirkpatrick and Przytycka~\cite{KP90} gave an $\Oh{\log n}$-time, $\Oh{n / \log n}$-processor algorithm for integer weights in the CREW PRAM model.  In another previous paper~\cite{Gag??} we used a data structure due to Kirkpatrick and Przytycka and a technique for generalized selection due to Klawe and Mumey~\cite{KM95}, to make Kirkpatrick and Klawe's algorithm for real weights run in $\Oh{\rule{0ex}{2ex} n \min (\log n, d \log \log n)}$ time, where $d$ is the number of distinct values $\lceil w_i \rceil$.  In this paper we prove a conjecture we made then, that a similar modification can make Drmota and Szpankowski's algorithm run in $\Oh{n}$ time.

\section{Applications} \label{sec:apps}

In the full version of this paper, we will consider all of the following problems:
\begin{enumerate}
\renewcommand{\theenumi}{\Alph{enumi}}
\item \label{prob:ptwise}
    build a prefix code with minimum maximum pointwise redundancy;
\item \label{prob:code}
    given a good estimate of the distribution over an alphabet, build a good prefix code;
\item \label{prob:test}
    given a good estimate of the distribution over a set, design a good group test to find the unique target;
\item \label{prob:reals}
    build a minimax tree for a multiset of real weights;
\item \label{prob:shannon}
    build a Shannon code;
\item \label{prob:depths}
    build a tree whose leaves have at most given depths;
\item \label{prob:circuit}
    restrict a circuit to have bounded fan-in or fan-out;
\item \label{prob:ints}
    build a minimax tree for a multiset of integer weights.
\end{enumerate}
The authors cited in the introduction have already shown, however, that Problem~\ref{prob:ptwise} takes $\Oh{n}$ more time than \ref{prob:reals}, \ref{prob:shannon} than \ref{prob:depths}, and \ref{prob:depths} and \ref{prob:circuit} than \ref{prob:ints}.  Therefore, in the current version of this paper, we consider only Problems~\ref{prob:code}, \ref{prob:test}, \ref{prob:reals} and \ref{prob:ints}.  In the remainder of this section we define what we mean by ``good'' in Problems~\ref{prob:code} and \ref{prob:test}, and show they take $\Oh{n}$ more time than \ref{prob:reals}.  Problems~\ref{prob:code} and \ref{prob:test} are, in fact, equivalent to each other and to \ref{prob:ptwise}, and analogous to a problem we considered in our paper~\cite{Gag??} on building alphabetic minimax trees.  In Section~\ref{sec:ints} we give two $\Oh{n}$-time algorithms for Problem~\ref{prob:ints}.  Finally, in Section~\ref{sec:reals} we show how to use either of those algorithms to obtain an algorithm for Problem~\ref{prob:reals} that takes $\Oh{n}$ time on a word RAM with \(\Omega (\log n)\)-bit words.  It follows that all the problems listed above take $\Oh{n}$ time.

Suppose we want to build a good prefix code with which to compress a file, but we are given only a sample of its characters.  Let \(P = p_1, \ldots, p_n\) be the normalized distribution of characters in the file, let \(Q = q_1, \ldots, q_n\) be the normalized distribution of characters in the sample and suppose our codewords are \(C = c_1, \ldots, c_n\).  An ideal code for $Q$ assigns the $i$th character a codeword of length \(\log (1 / q_i)\) (which may not be an integer), and the average codeword's length using such a code is \(H (P) + D (P \| Q)\), where \(H (P) = \sum_i p_i \log (1 / p_i)\) is the entropy of $P$ and \(D (P \| Q) = \sum_i p_i \log (p_i / q_i)\) is the relative entropy between $P$ and $Q$.  The entropy measures our expected surprise at a character drawn uniformly at random from the file, given $P$; the relative entropy (also known as the informational divergence or Kullback-Leibler pseudo-distance) measures the increase in our expected surprise when we estimate $P$ by $Q$, and is often used to quantify how well $Q$ approximates $P$ (see, e.g.,~\cite{CT06}).

Consider the best worst-case bound we can achieve, given only $Q$, on how much the average codeword's length exceeds \(H (P) + D (P \| Q)\).  A result by Katona and Nemetz~\cite{KN76} implies we do not generally achieve a constant bound on the difference when $C$ is a Huffman code for $Q$.  (Given $P$, of course, the best bound we could achieve on how much the average codeword's length exceeds \(H (P)\), would be the redundancy of a Huffman code for $P$.)  For example, if \(q_1, \ldots, q_n\) are proportional to \(F_n, \ldots, F_1\), where $F_i$ denotes the $i$th Fibonacci number (i.e., \(F_1 = F_2 = 1\) and \(F_i = F_{i - 1} + F_{i - 2}\) for \(i \geq 3\)), then the codewords' lengths are \(1, \ldots, n - 2, n - 1, n - 1\) in any Huffman code for $Q$.  If $p_n$ is sufficiently close to 1, then
\begin{eqnarray*}
\lefteqn{H (P) + D (P \| Q)}\\
& \approx & \log (1 / q_n)\\
& = & \log \sum_{i = 1}^n F_i\\
& = & n \log \phi + \Oh{1}
\end{eqnarray*}
but the average codeword's length \(\sum_i p_i |c_i| \approx n - 1\), so for large $n$ the difference is about \((1 / \log \phi - 1) n \approx 0.44 n\), where \(\phi \approx 1.62\) is the golden ratio.

As long as \(q_i > 0\) whenever \(p_i > 0\), the average codeword's length
\begin{eqnarray*}
\sum_i p_i |c_i|
& = & \sum_i p_i \left( \rule{0ex}{2.5ex}
    \log (1 / p_i) + \log (p_i / q_i) + \log q_i + |c_i| \right)\\
& = & H (P) + D (P \| Q) + \sum_i p_i (\log q_i + |c_i|)
\end{eqnarray*}
(if \(q_i = 0\) but \(p_i > 0\) for some $i$, then \(D (P \| Q)\) is infinite). Notice each $|c_i|$ is the length of a branch in the code-tree for $C$. Therefore, the best bound we can achieve is
\begin{eqnarray*}
\lefteqn{\min_C \max_P \left\{ \sum_i p_i (\log q_i + |c_i|) \right\}}\\
& = & \min_C \max_i \{\log q_i + |c_i|\}\\
& = & M (\log q_1, \ldots, \log q_n)\,,
\end{eqnarray*}
which is less than 1, by inspection of Drmota and Szpankowski's algorithm.  (Recall that \(M (\log q_1, \ldots, \log q_n)\) denotes the minimax cost of \(\{\log q_1, \ldots, \log q_n\}\), i.e., the weight of the root of a minimax tree for \(\{\log q_1, \ldots, \log q_n\}\).)  Moreover, we achieve this bound when the code-tree for $C$ has the same shape as a minimax tree for \(\{\log q_1, \ldots, \log q_n\}\).  In other words, Problem~\ref{prob:code} takes $\Oh{n}$ more time than \ref{prob:reals}.

Now suppose we want to design a good group test (see, e.g.,~\cite{AW87,Aig88}) to find the unique target in a set, given only an estimate $Q$ --- presumably gained from past experience or experimentation --- of the probability distribution $P$ according to which the target is chosen.  A group test allows us to choose, repeatedly, a subset of the elements and check whether the target is among them.  We can represent a group test as a decision tree in which each leaf is labelled with an element and each internal node is labelled with the concatenation of its children's labels.  Because such a decision tree can be viewed as the code-tree for a prefix code, and vice versa, the expected number of checks we make exceeds \(H (P) + D (P \| Q)\) by as little as possible when the decision tree for our group test has the same shape as a minimax tree for \(\{\log q_1, \ldots, \log q_n\}\).  In other words, Problem~\ref{prob:test} is equivalent to \ref{prob:code} and, therefore, also takes $\Oh{n}$ more time than \ref{prob:reals}.

We are currently studying whether either Drmota and Szpankowski's solution to Problem~\ref{prob:ptwise} or our solution to~\ref{prob:code} can give us an intuitive explanation of why dynamic Shannon coding has a better worst-case bound than dynamic Huffman coding does.  On the one hand, worst-case bounds (especially for online algorithms; see, e.g.,~\cite{BE98}) are often proven by considering a game between the algorithm and an omniscient adversary, and minimizing the maximum pointwise redundancy at each step seems somehow related (more than just by name) to the minimax strategy for the algorithm.  On the other hand, dynamic prefix coding can be viewed as a procedure in which we repeatedly build a prefix code based on a sample --- i.e., the characters already encoded.

\section{Minimax Trees for Integer Weights} \label{sec:ints}

In this section we give two $\Oh{n}$-time algorithms for building a minimax tree for a multiset of integer weights, both based on the following lemma (which we note applies to any weights, not only integers) and corollary:

\begin{lemma} \label{lem:replacement}
If \(W = \{w_1, \ldots, w_n\}\) is a multiset of weights and
\[\textstyle W'
= \left\{ \rule{0ex}{2.5ex}
    \max \left( \rule{0ex}{2ex} w_1, \max_i \{w_i\} - n + 1 \right), \ldots,
    \max \left( \rule{0ex}{2ex} w_n, \max_i \{w_i\} - n + 1 \right) \right\}\,,\]
then \(M (W') = M (W)\).  Moreover, any minimax tree for $W'$ becomes a minimax tree for $W$ when we replace the leaves' weights equal to \(\max_i \{w_i\} - n + 1\) by the weights in $W$ less than or equal to \(\max_i \{w_i\} - n + 1\), in any order.
\end{lemma}

\begin{proof}
Consider a minimax tree $T$ for $W$.  Without loss of generality, we can assume $T$ is strictly binary --- i.e., that every internal node has exactly two children --- and, therefore, that it has height at most \(n - 1\).  (Recall that, for simplicity, we consider only binary trees.)  If \(n = 1\), then \(W = w_1 = \max_i \{w_i\} - n + 1\).  Otherwise, all the leaves have depth at least 1, so \(M (W) \geq \max_i \{w_i\} + 1\).  Consider any leaf (if one exists) with weight less than \(\max_i \{w_i\} - n + 1\) and depth $\ell$.  Since \(\max_i \{w_i\} - n + 1 + \ell \leq \max_i \{w_i\} < M (W)\), increasing that leaf's weight to \(\max_i \{w_i\} - n + 1\) and updating its ancestors' weights, does not change the weight \(M (W)\) of the root.  It follows that \(M (W') = M (W)\).

Now consider a minimax tree $T'$ for $W'$.  If we replace the leaves' weights equal to \(\max_i \{w_i\} - n + 1\) by the weights in $W$ less than or equal to \(\max_i \{w_i\} - n + 1\) and update all the nodes' weights, then the weight \(M (W')\) of the root cannot increase nor, by definition, decrease to less than \(M (W)\).  Since \(M (W') = M (W)\), it follows that the re-weighted tree is a minimax tree for $W$. \qed
\end{proof}

\begin{corollary} \label{cor:sorting}
When all the weights in $W$ are integers, we can sort $W'$ in $\Oh{n}$ time.
\end{corollary}

\begin{proof}
When all the weights in $W$ at least \(\max_i \{w_i\} - n + 1\) are integers, all the weights in $W'$ are integers in the interval \(\left[ \rule{0ex}{2ex} \max_i \{w_i\} - n + 1, \max_i \{w_i\} \right]\).  Since this interval has length \(n - 1\), we can sort $W'$ in $\Oh{n}$ time using either direct addressing, which takes $\Oh{n}$ extra space, or radix sort, which takes no extra space~\cite{FMP07}. \qed
\end{proof}

For our first algorithm, we build and sort $W'$; build a minimax tree for $W'$ using a implementation of Golumbic's algorithm that takes $\Oh{n}$ time when the weights are already sorted; and replace the leaves' weights equal to \(\max_i \{w_i\} - n + 1\) by the weights in $W$ less than or equal to \(\max_i \{w_i\} - n + 1\).  We note that Van Leeuwen~\cite{VLe76} showed how to implement Huffman's algorithm to take $\Oh{n}$ when the weights are already sorted.  We could implement Golumbic's algorithm analogously, but we think the implementation below is simpler.

\begin{lemma} \label{lem:sorted}
Golumbic's algorithm can be implemented to take $\Oh{n}$ time when the weights are already sorted.
\end{lemma}

\begin{proof}
We start with the weights stored in a linked list in nondecreasing order, and set a pointer to the head of the list.  We then repeat the following procedure until there is only one node left in the list, which is the root of a minimax tree for the given weights: we move the pointer along the list to the last weight less than or equal to the maximum of the first two weights plus 1; remove the first two nodes from the list; make those nodes the children of a new node with weight equal to the maximum of their weights plus one; and insert the new node immediately to the right of the pointer.  Notice we remove two nodes for each one we insert, so the total number of nodes is \(2 n - 1\).  Therefore, since the pointer passes over each node once, this implementation takes $\Oh{n}$ time.
\qed
\end{proof}

\noindent Building and sorting $W'$ takes $\Oh{n}$ time, by Corollary~\ref{cor:sorting}; building a minimax tree for $W'$ takes $\Oh{n}$ time, by Lemma~\ref{lem:sorted}; replacing the leaves' weights equal to \(\max_i \{w_i\} - n + 1\) by the weights in $W$ less than or equal to \(\max_i \{w_i\} - n + 1\) takes $\Oh{n}$ time, because it can be done in any order.  By Lemma~\ref{lem:replacement}, the resulting tree is a minimax tree for $W$.

\begin{theorem} \label{thm:ints}
Given a multiset $W$ of $n$ integer weights, we can build a minimax tree for $W$ in $\Oh{n}$ time.
\end{theorem}

Our second algorithm differs in its second step: instead of using Golumbic's algorithm to build a minimax tree for $W'$, we use Kirkpatrick and Klawe's $\Oh{n}$-time algorithm for integer weights to build an alphabetic minimax tree for the sequence $V$ consisting of the weights in $W'$ in non-increasing order.  The algorithm's correctness follows from the Kraft Inequality:

\begin{theorem}[Kraft, 1949] \label{thm:kraft}
If there exists a binary tree whose leaves have depths \(\ell_1, \ldots, \ell_n\), then \(\sum_i 1 / 2^{\ell_i} \leq 1\).  Conversely, if \(\sum_i 1 / 2^{\ell_i} \leq 1\) and \(\ell_1 \leq \cdots \leq \ell_n\), then there exists an ordered binary tree whose leaves, from left to right, have depths \(\ell_1, \ldots, \ell_n\).
\end{theorem}

\noindent By the latter part of Theorem~\ref{thm:kraft} and a standard exchange argument --- i.e., if a minimax tree contains two leaves such that the deeper one has a higher weight than the shallower one, then we can swap their weights --- there exists a minimax tree for $W'$ in which the leaves' weights are non-increasing from left to right.  Therefore, by definition, any alphabetic minimax tree for $V$ is a minimax tree for $W'$.

\section{Minimax Trees for Real Weights} \label{sec:reals}

Strictly speaking, Drmota and Szpankowski's algorithm works only when given a multiset of weights equal to \(\{\log p_1, \ldots, \log p_n\}\) for some probability distribution \(P = p_1, \ldots, p_n\).  For any value $c$, however, if \(W = \{w_1, \ldots, w_n\}\) and \(W' = \{w_1 + c, \ldots, w_n + c\}\) then, by definition, \(M (W') = M (W) + c\) and any minimax tree for $W'$ becomes a minimax tree for $W$ when we subtract $c$ from each leaf's weight.  In particular, if \(c = - \log \left( \sum_i 2^{w_i} \right)\) then \(\sum_i 2^{w_i + c} = 2^c \sum_i 2^{w_i} = 1\); therefore, \(W' = \{\log p_1, \ldots, \log p_n\}\) for some probability distribution \(P = p_1, \ldots, p_n\) and we can use Drmota and Szpankowski's algorithm to build minimax trees for $W'$ and, thus, for $W$.  Without loss of generality, we henceforth assume the given multiset $W$ of weights is equal to \(\{\log p_1, \ldots, \log p_n\}\) for some probability distribution $P$ (so each \(w_i \leq 0\)).

\begin{theorem}[Drmota and Szpankowski, 2002] \label{thm:D&S}
If \(W = \{w_1, \ldots, w_n\}\) is a multiset of weights, \(X = \{x_1, \ldots, x_n\} = \{|w_1| - \lfloor |w_1| \rfloor, \ldots, |w_n| - \lfloor |w_n| \rfloor\}\) and $x_i$ is the largest element in \(X \cup \{0\}\) such that
\[\sum_{x_j \leq x_i} 1 / 2^{\lfloor |w_j| \rfloor} +
    \sum_{x_j > x_i} 1 / 2^{\lceil |w_j| \rceil}
\leq 1\,,\]
then any minimax tree for \(\{- \lfloor |w_j| \rfloor\,:\,x_j \leq x_i\} \cup \{- \lceil |w_j| \rceil\,:\,x_j > x_i\}\) becomes a minimax tree for $W$ when we replace each leaf's weight $- \lfloor |w_j| \rfloor$ or $- \lceil |w_j| \rceil$ by $w_j$.
\end{theorem}

If \(x_1 \leq \cdots \leq x_n\) and \(x_i > 0\) then, by Theorem~\ref{thm:D&S}, $i$ is the largest index such that \(\{\lfloor |w_j| \rfloor\,:\,x_j \leq x_i\} \cup \{\lceil |w_j| \rceil\,:\,x_j > x_i\}\) satisfies the Kraft Inequality.  To build a minimax tree for $W$ with Drmota and Szpankowski's algorithm, we compute and sort $X$; use binary search to find $i$, in each round testing whether the Kraft Inequality holds; build a minimax tree for \(\{- \lfloor |w_1| \rfloor, \ldots, - \lfloor |w_i| \rfloor, - \lceil |w_{i + 1}| \rceil,\) \(\ldots, - \lceil |w_n| \rceil\}\); and replace each leaf's weight $- \lfloor |w_j| \rfloor$ or $- \lceil |w_j| \rceil$ by $w_j$.  Our version differs in three ways: we use generalized selection instead of sorting and binary search; we use a new data structure to test the Kraft Inequality; and we use either of our algorithms from Section~\ref{sec:ints} to build the minimax tree for \(\{- \lfloor |w_1| \rfloor, \ldots, - \lfloor |w_i| \rfloor, - \lceil |w_{i + 1}| \rceil, \ldots, - \lceil |w_n| \rceil\}\).  In the remainder of this section we first show how to use generalized selection to find $i$ in $\Oh{n}$ time, excluding the time needed to test the Kraft Inequality; we then show how to perform all the necessary tests in a total of $\Oh{n}$ time on a word RAM with \(\Omega (\log n)\)-bit words, using our new data structure.  Since each of our algorithms from Section~\ref{sec:ints} takes $\Oh{n}$ time, it follows that we can build a minimax tree for $W$ in $\Oh{n}$ time.

To find $x_i$ in $\Oh{n}$ time with general selection, we start with the multiset \(X_1 = X \cup \{0\}\) and repeat the following procedure until we reach the empty set: in the $r$th round, we use the linear-time selection algorithm due to Blum {\it et al.}~\cite{BFP+73} to find the current multiset $X_r$'s median $x_m$, then test whether
\[\sum_{x_j \leq x_m} 1 / 2^{\lfloor |w_j| \rfloor} +
    \sum_{x_j > x_m} 1 / 2^{\lceil |w_j| \rceil}
\leq 1\,;\]
if so, we remove those elements of $X_r$ that are less than or equal to $x_m$ and recurse on the resulting multiset; if not, we remove those elements of $X_r$ that are greater than or equal to $x_m$ and recurse.  The element $x_i$ is the largest median we consider for which the test is positive.  Since the size of the multisets decreases by a factor of at least 2 in each round, we use $\Oh{\log n}$ rounds and we find all the medians in a total of $\Oh{n}$ time.

By the same arguments we used to prove Lemma~\ref{lem:replacement}, we can assume, without loss of generality, that \(\lceil |w_j| \rceil \leq n - 1\) for each $j$.  To test the Kraft Inequality, we use a data structure consisting of two $n$-bit binary fractions, $S_1$ and $S_2$, each broken into \((\log n)\)-bit blocks and initially set to 0.  For \(1 \leq k \leq n - 1\), adding \(1 / 2^k\) to either fraction takes $\Oh{1}$ amortized time, for the same reason that incrementing a binary counter takes $\Oh{1}$ amortized time (see, e.g.,~\cite[Section 17.3]{CLRS01}).  On a word RAM with \(\Omega (\log n)\)-bit words, nondestructively testing whether \(S_1 + S_2 \leq 1\) takes $\Oh{n / \log n}$ time, because adding each corresponding pair of blocks takes $\Oh{1}$ time and, by induction, the number carried from each pair to the next is at most 1; resetting either fraction to 0 takes $\Oh{1}$ time for each block, i.e., $\Oh{n / \log n}$ time in total.

Before starting to search for $x_i$, we set \(S_1 = \sum_j 1 / 2^{\lceil |w_j| \rceil}\) in $\Oh{n}$ time.  Throughout our generalized selection, we maintain the invariant that, at the beginning of the $r$th round,
\[S_1 = \sum_j 1 / 2^{\lceil |w_j| \rceil} +
    \sum_{0 < x_j < \min (X_r)} 1 / 2^{\lceil |w_j| \rceil}\]
and \(S_2 = 0\).  In the $r$th round, we set
\[S_2 = \sum_{\min (X_r) \leq x_j \leq x_m} 1 / 2^{\lceil |w_j| \rceil}\]
in $\Oh{|X_r|}$ time.  Since
\begin{eqnarray*}
S_1 + S_2
& = & \sum_j 1 / 2^{\lceil |w_j| \rceil} +
    \sum_{0 < x_j < \min (X_r)} 1 / 2^{\lceil |w_j| \rceil} +
    \sum_{\min (X_r) \leq x_j \leq x_m} 1 / 2^{\lceil |w_j| \rceil}\\
& = & \sum_{x_j \leq x_m} 1 / 2^{\lfloor |w_j| \rfloor} +
    \sum_{x_j > x_m} 1 / 2^{\lceil |w_j| \rceil}\,,
\end{eqnarray*}
we can test the Kraft Inequality in $\Oh{n / \log n}$ time by checking whether \(S_1 + S_2 \leq 1\).  If the test is positive, then we add $S_2$ to $S_1$ in $\Oh{n / \log n}$ time; if the test is negative, then we do not change $S_1$.  In either case, straightforward calculation shows that, afterwards,
\[S_1 = \sum_j 1 / 2^{\lceil |w_j| \rceil} +
    \sum_{0 < x_j < \min (X_{r +1})} 1 / 2^{\lceil |w_j| \rceil}\,\]
so the first part of our invariant is maintained.  Finally, we reset \(S_2 = 0\) in $\Oh{n / \log n}$ time, so the second part of our invariant is maintained.  Since \(|X_r| = \Oh{n / 2^r}\), the $r$th round takes a total of $\Oh{n / 2^r + n / \log n}$ time.  Since \(\sum_{r \geq 1} n / 2^r = n\) and we use $\Oh{\log n}$ rounds, it follows that our whole generalized selection takes $\Oh{n}$ time.  This completes the proof of our main result:

\begin{theorem} \label{thm:reals}
Given a multiset $W$ of $n$ real weights, we can build a minimax tree for $W$ in $\Oh{n}$ time on a word RAM with \(\Omega (\log n)\)-bit words.
\end{theorem}

\bibliographystyle{plain}
\bibliography{iwoca}

\end{document}